\documentclass[a4paper,10pt]{article}

\usepackage{amsmath}
\usepackage{amssymb}
\usepackage{amsthm}
\usepackage{graphicx}
\usepackage{tikz}
\usepackage{setspace}
\usepackage[nottoc,numbib]{tocbibind}
\usepackage[font=small,labelfont=bf]{caption}

\newtheorem{theorem}{Theorem}[section]
\newtheorem{proposition}[theorem]{Proposition}
\newtheorem{conjecture}[theorem]{Conjecture}

\newtheorem{corollary}[theorem]{Corollary}
\newtheorem{lemma}[theorem]{Lemma}
\newtheorem{example}[theorem]{Example}

\theoremstyle{definition}
\newtheorem{definition}[theorem]{Definition}

\theoremstyle{remark}
\newtheorem{remark}[theorem]{Remark}

\begin{document}
\pagenumbering{arabic}
\title{Asymptotic Linear Programming Lower Bounds for the Energy of
Minimizing Riesz and Gauss Configurations }
\author{D. P. Hardin, T. J. Michaels, and E.B. Saff\footnote{ The research of the authors was supported, in part, by National Science Foundation grant DMS-1516400. The research of T. Michaels was completed as part of his Ph.D. dissertation at Vanderbilt University. Research for this article was conducted while two of the authors were in residence at the Institute for Computational and Experimental Research in Mathematics in Providence, RI, during the ``Point Configurations in Geometry, Physics and Computer Science" program supported by the National Science Foundation under Grant No. DMS-1439786.}}

\maketitle
\begin{abstract}
Utilizing frameworks developed by Delsarte, Yudin and Levenshtein, we deduce linear programming lower bounds (as $N\to \infty$) for
the Riesz energy of $N$-point configurations on the $d$-dimensional unit sphere in the so-called hypersingular case; i.e, for non-integrable
Riesz kernels of the form $|x-y|^{-s}$ with $s>d.$ As a consequence, we immediately get  (thanks to the Poppy-seed bagel theorem) lower estimates for the large $N$ limits of minimal hypersingular Riesz energy on compact $d$-rectifiable sets. Furthermore, for the Gaussian potential $\exp(-\alpha|x-y|^2)$ on $\mathbb{R}^p,$ we obtain
lower bounds for the energy of infinite configurations having a prescribed density.
\end{abstract}

\section{Introduction}

Minimal energy configurations have wide ranging applications in various scientific fields such as cryptography, crystallography, viral morphology, as well as in finite element modeling, radial basis functions, and Quasi-Monte-Carlo methods for graphics applications. For a fixed dimension and cardinality, the use of the Delsarte-Yudin linear programming bounds and Levenshtein $1/N$-quadrature rules are known to provide bounds on the minimal energy and prove universal optimality of some configurations on the sphere $\mathbb{S}^d$ (see for example \cite{CohnKumar}). The goal of this paper is to adapt these techniques to provide lower bounds on minimal energy for configurations in two different but related contexts.  The first is for the large $N$ limit of Riesz energy of $N$-point configurations  on a compact $d$-rectifiable set embedded in $\mathbb{R}^p$, while the second is for the Gaussian energy of infinite configurations in $\mathbb{R}^p$ having a prescribed density. The latter provides an alternative method for obtaining
a main result of Cohn and de Courcy-Ireland \cite{CohnIre}.
%in  In this case, the potentials are dominated by short range interactions, and the limit in some dimensions is conjectured to be intimately related to the best packing problem as we describe.

For our results on Riesz potentials we need the following definitions and notations. We say a set $A\subset \mathbb{R}^p$ is \textit{$d$-rectifiable} if it is the image of a bounded set in $\mathbb{R}^d$ under a Lipschitz mapping. For a $d$-rectifiable, closed set $A$ and a lower semicontinuous, symmetric kernel $K: A\times A\to \ (-\infty,\infty]$, the $K$-energy of a configuration $\omega_N = \left\{x_1,\ldots,x_N\right\}\subset A$ of $N$ (not necessarily distinct) points is given by
\[E_K(\omega_N):= \sum_{i\neq j} K(x_i,x_j).\]
A commonly arising problem is to minimize the $K$-energy for a fixed number of points and describe the optimal configurations; i.e., to determine
\[\mathcal{E}_K(A,N):= \inf_{\omega_N\subset A} E_K(\omega_N).\]\

For point configurations on compact sets we will primarily focus on the Riesz $s$-kernels
$$K_s(x,y):= |x-y|^{-s}\,\, \textrm{for}\,\, s>d=\dim(A);$$
that is, in the \emph{hypersingular case}, which is intimately related to the best-packing problem. We remark that for such hypersingular
kernels, the \textit{continuous $s$-energy} of $A$
\[ \mathcal{I}_s[\mu]:=\int_{A} \int_{A} K_s(x,y) d\mu(x)d\mu(y)\]
is infinite for every probability measure $\mu$ supported on $A,$ and so the standard methods of potential theory for obtaining large $N$ limits
of minimizing point configurations do not apply.\

 For brevity we hereafter set
$$ E_s(\omega_N):=E_{K_s}(\omega_N),\qquad \mathcal{E}_s(A,N):=\mathcal{E}_{K_s}(A,N).
$$ Furthermore, if $A$ is the unit sphere $\mathbb{S}^d \subset \mathbb{R}^{d+1}$ and
 $K(x,y)$ is a kernel on $ \mathbb{S}^d \times \mathbb{S}^d$ of the form $K(x,y)=h(\langle x, y\rangle)$ for some function $h$ on $[-1,1]$, we write
$$ E_h(\omega_N)=E_K(\omega_N),\qquad \mathcal{E}_h(\mathbb{S}^d,N)=\mathcal{E}_K(\mathbb{S}^d,N).
$$
In particular,
$$K_s(x,y)=h_s(\langle x, y\rangle):=(2-2\langle x, y\rangle)^{-s}.$$\

For fixed cardinalities $N$ and kernels of the form $K(x,y)=h(\langle x, y\rangle)$, a general framework for obtaining lower bounds for minimal energy configurations on the  unit sphere  was developed by Yudin \cite{Yudin} based on a method of Delsarte, Goethals, and Seidel \cite{DGS} for spherical designs. This linear programming technique involves maximizing a certain functional defined over a constrained class of functions $f$ that satisfy $f(t)\leq h(t)$ for $t\in [-1,1].$ Combining Yudin's approach with Levenshtein's work \cite{LevBig},\cite{LevPacking} on maximal spherical codes, Boyvalenkov et al \cite{Petersquared} derived lower bounds for discrete energy that are `universal' in the sense that they hold whenever the potential function $h(t)$ is \emph{absolutely monotone} on $[-1,1);$ that is, when $h^{(k)}(t)$ exists and is non-negative for $t \in [-1,1)$ for all $k \geq 0,$ and $h(1):=\lim_{t\to 1^{-}}h(t)$, which may be $+\infty.$

In the present paper, we use this framework to derive asymptotic lower bounds as $N \to \infty$ for $\mathcal{E}_s(\mathbb{S}^d,N)$ in the case $s>d.$ These results for the sphere, in turn, have application to
the broader class of energy problems on $d$-rectifiable sets. Indeed, this is a  consequence of the localized nature of the potentials $h_s$ as expressed in the following result, which is known as the \emph{Poppy-seed bagel theorem}.

\begin{theorem}[\cite{CSD}, \cite{borharsaf}]

For any $d$-rectifiable closed set $A\subset \mathbb{R}^p$ and any $s>d$, there exists a positive, finite constant $C_{s,d}$, independent of $A$ such that

\begin{equation}\label{pop}
\lim_{N\to\infty}\frac{\mathcal{E}_s(A,N)}{N^{1+s/d}} = \frac{C_{s,d}}{\mathcal{H}_d(A)^{s/d}}.
\end{equation}\label{poppy seed}
\noindent Furthermore, any sequence of $N$-point $s$-energy minimizing configurations is asymptotically uniformly
distributed with respect to $d$-dimensional Hausdorff measure restricted to $A$.
\end{theorem}
In \eqref{pop}, $\mathcal{H}_d(A)$ denotes the $d$-dimensional Hausdorff measure of $A$ with the normalization
that the $d$-dimensional unit cube embedded in $\mathbb{R}^p$ has measure 1.\

In dimension $d=1$, it is known \cite{CS1} that $C_{s,1} = 2\zeta(s)$,  but for all other dimensions the exact values of $C_{s,d}$ have not as yet been proven. However, the following relation between $C_{s,d}$ and the optimal packing density in $\mathbb{R}^d$ was established in \cite{sinfty}:
\begin{equation}
\lim_{s\to\infty} [C_{s,d}]^{1/s} = \frac{1}{C_{\infty,d}},\ \ \ \ \ \ \ \ \ \ \ \ C_{\infty,d}:= 2\bigg[\frac{\Delta_d}{\mathcal{H}_d(\mathbb{B}^d)}\bigg]^{1/d},
\label{eq.Csinfty}
\end{equation}
where $\Delta_d$ is the largest sphere packing density in $\mathbb{R}^d$. The only dimensions for which $\Delta_d$ is known at present are $d=1,2,3$ and, more recently, $d=8$ and $d=24$ (see \cite{Via8} and \cite{Via24}). In these special dimensions, $\Delta_d$ is attained by lattice packings, which is not expected to be the case for general dimensions.

Clearly, any sequence of configurations on a set $A$ provides an upper bound for $C_{s,d}$. Furthermore, it is straightforward (see, for example, \cite{BHS12}, Proposition 1) to establish that

\begin{equation}
    C_{s,d}\leq \min_{\Lambda\subset\mathbb{R}^d}|\Lambda|^{s/d}\zeta_\Lambda(s),
\label{eq.lattice bound}
\end{equation}
where the minimum is taken over all lattices $\Lambda\subset\mathbb{R}^d$ with covolume $|\Lambda|>$0 and
\begin{equation}
    \zeta_\Lambda(s):=\sum_{0\neq x\in\Lambda} |x|^{-s}
\label{eq.epszeta}
\end{equation}
is the Epstein zeta function for the lattice. Regarding equality, the following conjecture is well known \cite{BHS12}, \cite{CohnKumar}:

\begin{conjecture}
For $d = 2,4,8,$ and $24$,
\[C_{s,d} = \widetilde{C}_{s,d}:=|\Lambda_d|^{\frac{s}{d}}\zeta_{\Lambda_d}(s),\qquad s>d,\]
where $\Lambda_2$ is the equi-triangular lattice, $\Lambda_4$ the $D_4$ lattice, $\Lambda_8$ the $E_8$ lattice, and $\Lambda_{24}$ the Leech lattice.
\label{Csdconj}
\end{conjecture}

General lower bounds on $C_{s,d}$ have been less studied. A crude but simple lower bound arises from the following convexity argument (cf. \cite{kuisaf}).\

 Let $\omega_N^*=\left\{x_1,\ldots, x_N\right\}$ be a minimizing $N$-point
$s$-energy configuration on $\mathbb{S}^d$ and, for each $i = 1,\ldots, N$, let $\delta_i:=\min_{j\neq i} |x_i-x_j|$. With $C(x_i, \delta_i/2)$ denoting the spherical cap with center $x_i$ and Euclidean radius $\delta_i/2$, we deduce that  $\sum_{i=1}^N \mathcal{H}_d(C(x_i, \delta_i/2)) \leq\mathcal{H}_d(\mathbb{S}^d).$ It is easily verified that
 \begin{equation}\label{HdCap}
 \mathcal{H}_d(C(x_i, r))=\mathcal{H}_d(S^d) \frac{r^d}{\lambda_dd}+\mathcal{O}(r^{d+2}),\,\, \,\, r\to 0^{+},
 \end{equation}
 where
 \begin{equation}\lambda_d := \int\limits_{-1}^1(1-t^2)^{\frac{d-2}{2}} dt = \frac{\sqrt{\pi}\Gamma(\frac{d}{2})}{\Gamma(\frac{d+1}{2})}.
\label{eq.lambdad}
\end{equation}
Thus for $1>\epsilon>0$ and all $N$ sufficiently large we have from the asymptotic denseness of
the minimizing configurations (Theorem \ref{poppy seed}) that
$$(1-\epsilon)\mathcal{H}_d(S^d) \frac{1}{2^d\lambda_dd}\sum_{i=1}^N \delta_i^d \leq \sum_{i=1}^N \mathcal{H}_d(C(x_i, \delta_i/2)) \leq\mathcal{H}_d(\mathbb{S}^d)
$$
and so
\begin{equation}\label{upper}
\sum_{i=1}^N \delta_i^d \leq (1-\epsilon)^{-1}2^d\lambda_dd.
\end{equation}
By convexity, we also have
\[\mathcal{E}_s(\mathbb{S}^d,N)=\sum_{i\neq j}\frac{1}{|x_i-x_j|^s}\geq \sum_{i=1}^N \frac{1}{\delta_i^s}= \sum_{i=1}^N (\delta_i^d)^{-s/d}\geq N\bigg(\frac{1}{N}\sum_{i=1}^N \delta_i^d\bigg)^{-s/d}.\]
Consequently, from \eqref{upper} we obtain for $N$ large
$$\frac{\mathcal{E}_s(\mathbb{S}^d,N)}{N^{1+s/d}}\geq \left((1-\epsilon)^{-1}2^d\lambda_dd\right)^{-s/d}.
$$
Letting first $N\to\infty$ and then $\epsilon\to 0,$ Theorem \ref{poppy seed} yields the estimate

\begin{equation}
C_{s,d}\geq\Theta_{s,d}:=\left(\frac{\mathcal{H}_d(\mathbb{S}^d)}{2^{d}\lambda_dd}\right)^{s/d}=\frac{1}{2^s}
\left(\mathcal{H}_{d-1}(\mathbb{S}^{d-1})/d\right)^{s/d}.
\label{basicbound}
\end{equation}\

A less trivial lower bound is the following, established in \cite{BHS12}:
\begin{proposition}
If $d\geq 2$ and $s>d$, then for $(s-d)/2$ not an integer,
\[C_{s,d} \geq \xi_{s,d}:=\bigg[\frac{\pi^{d/2}\Gamma(1+\frac{s-d}{2})}{\Gamma(1+\frac{s}{2})}\bigg]^{s/d}\frac{d}{s-d}.\]
 \label{bhsbound}
\end{proposition}

Our main result for Riesz potentials is the following improvement over the lower bounds for $C_{s,d}$ in  \eqref{basicbound} and  Proposition \ref{bhsbound}.

\begin{theorem}  For a fixed dimension $d$, let $z_{i}$ be the $i$-th smallest positive zero of the Bessel function $J_{d/2}(z)$, $i = 1,2,\ldots .$ Then, for $s>d,$

\begin{equation}C_{s,d}\geq A_{s,d},
\label{eq.Asd}
\end{equation}
where
\begin{equation}
    A_{s,d}:=\bigg[\frac{\pi^{\frac{d+1}{2}}\Gamma(d+1)}{\Gamma(\frac{d+1}{2})}\bigg]^{s/d}\frac{4}{\lambda_d\Gamma(d+1)}\sum_{i=1}^\infty(z_{i})^{d-s-2}\big(J_{d/2+1}(z_{i})\big)^{-2}
\end{equation}
and $\lambda_d$ is as defined in \eqref{eq.lambdad}.
\label{thm.csd bound}
\end{theorem}
For $d=1$, $A_{s,d} = 2\zeta(s),$ which is optimal.  Furthermore, as we prove in Section~3,  both $A_{s,d}$ and $\xi_{s,d}$ have the same dominant behavior as $C_{s,d}$ as  $s\to d^{+}$; namely they all have a simple pole at $s=d$ with the same residue.   In Figure~\ref{FigBnds}, we compare
the bounds $A_{s,2}$, $\Theta_{s,2}$, $\xi_{s,2}$ with the conjectured value $\widetilde{C}_{2,s}$.
 \begin{figure}[htbp]
\begin{center}
\includegraphics[width=4in]{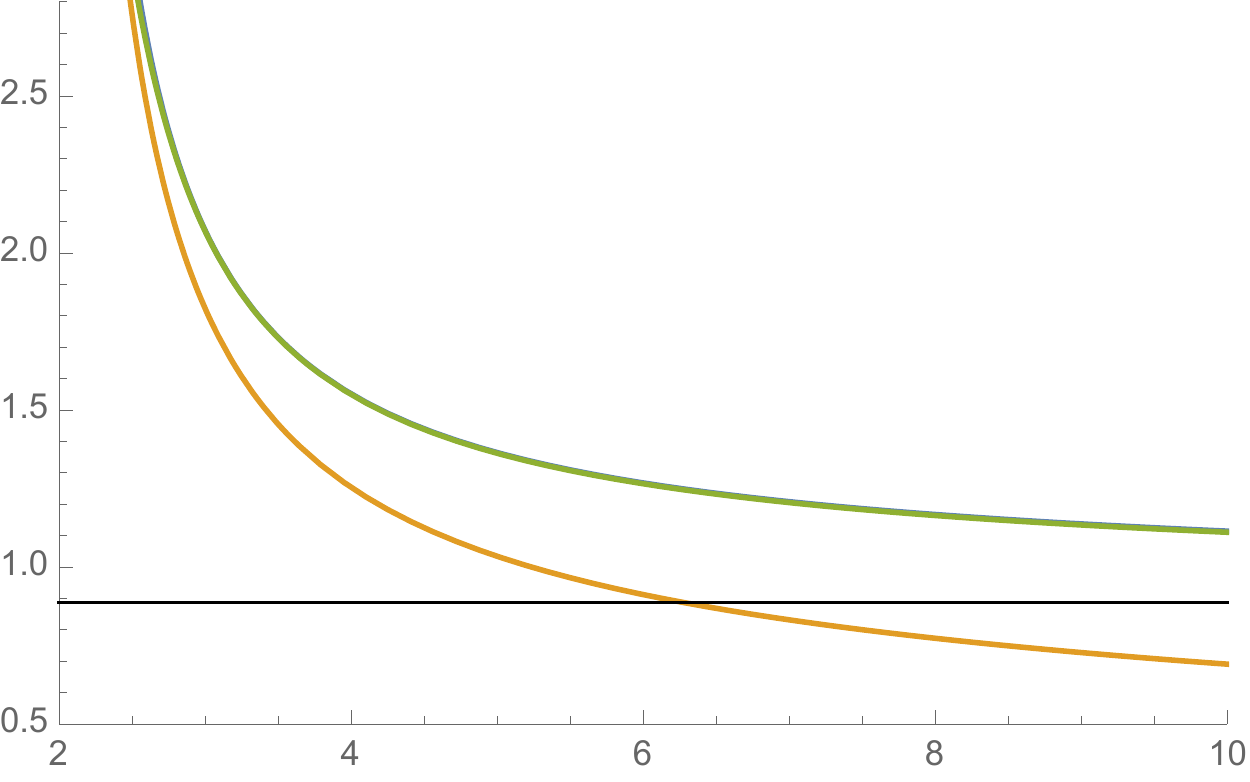}
\caption{{The two lower graphs show  $\Theta_{s,2}^{1/s}$ (constant graph) and $\xi_{s,2}^{1/s}$ while the upper two graphs (indistinguishable on this scale) show both $A_{s,2}^{1/s}$ and the conjectured value   $\widetilde{C}_{2,s}^{1/s}$ as $s$ ranges from $2$ to $10$.
}}
\label{FigBnds}
\end{center}
\end{figure}

\begin{proposition}Let $d\in {\mathbb N}$. Then
\begin{equation}
\lim_{s\to d^+} (s-d)C_{s,d} =\lim_{s\to d^+} (s-d)\xi_{s,d}=\lim_{s\to d^+} (s-d)A_{s,d}=  \frac{2\pi^{d/2}}{\Gamma(\frac{d}{2})}.
\label{eq.2tight}
\end{equation}
\label{thm.2tight}
\end{proposition}

As illustrated at the end of Section 2, the Levenshtein $1/N$-quadrature rules give bounds on the minimal separation distance for optimal packings on $\mathbb{S}^d$, and  $A_{s,d}$ recovers these bounds as $s\to\infty$. For $d = 2,4,8,$ and $24$, letting $\widetilde{C}_{s,d}$ be the conjectured values of $C_{s,d}$ from Conjecture \ref{Csdconj}, it is easy to verify that $\lim_{s\to\infty}(\widetilde{C}_{s,d}/A_{s,d})^{1/s}$ exists. Numerical comparisons between
$A_{s,d}$ and $\widetilde{C}_{s,d}$ are illustrated in Section 4.\

We next consider bounds for the Gaussian energy of infinite point configurations in $\mathbb{R}^d$. Our goal is to show that the method used to prove Theorem \ref{thm.csd bound} provides an alternative
approach to deriving the lower bounds obtained by Cohn and de Courcy-Ireland \cite{CohnIre}. We begin
with some essential definitions.

\begin{definition}\label{lowerfenergy}
For an infinite configuration $\mathcal{C}\subset\mathbb{R}^d$ and $f:(0,\infty)\to \mathbb{R}$, the \textit{lower f-energy of} $\mathcal{C}$ is

\[E_f(\mathcal{C}):=\liminf_{R\to\infty}\frac{1}{\#(\mathcal{C}\cap B^d(R))}\sum_{\substack{ x,y\in\mathcal{C}\cap B^d(R)\\ x\neq y}} f(|x-y|),\]
where $\#$ denotes cardinality and $B^d(R)$ is the $d$-dimensional ball of radius $R$ centered at $0$. If the limit exists, we call it the \textit{f-energy of} $\mathcal{C}$.
\end{definition}

\begin{definition}\label{lowerdensity}
The \textit{lower density} $\rho$ of a configuration $\mathcal{C}$ is defined to be
\[\rho:=\liminf_{R\to\infty}\frac{\#(\mathcal{C}\cap B^d(R))}{\textup{vol}(B^d(R))}.\]
If the limit exists, we call it the \textit{density of} $\mathcal{C}$.
\end{definition}

We shall show that universal lower bounds developed in \cite{Petersquared} and based on Delsarte-Levenshtein methods can be used to prove the
following estimate of Cohn and de Courcey-Ireland.  (The results in \cite{CohnIre} came to the authors' attention during the preparation of this manuscript  and appear in the dissertation of Michaels \cite{Mic2017}.)
\begin{theorem}[\cite{CohnIre}]
Let $f(|x-y|) = e^{-\alpha |x-y|^2},\,\alpha>0,$ be a Gaussian potential in $\mathbb{R}^d$ and choose $R_\rho$ so that \rm{vol}$(B^d(R_\rho/2))$ = $\rho$. \textit{Then the minimal f-energy for point configurations of density} $\rho$ \textit{in} $\mathbb{R}^d$ \textit{is bounded below by}

\begin{equation}
\frac{4}{\lambda_d\Gamma(d+1)}\sum_{i=1}^\infty z_i^{d-2}\big(J_{d/2+1}(z_{i})\big)^{-2} f\bigg(\frac{z_i}{\pi R_\rho}\bigg),
\label{eq.cohnire}
\end{equation}
\textit{where the} $z_i$'s \textit{are as in Theorem} \rm\ref{thm.csd bound}.
\label{thm.cohnire}
\end{theorem}
We remark that there is a strong relation connecting Theorems~\ref{thm.cohnire} and \ref{thm.csd bound}.  Indeed, 
if $f(r)=g(r^2)$ for some completely monotone function $g$ with sufficient decay, then there is some non-negative measure $\mu$ on $[0,\infty)$ such that (e.g., see \cite{Wid1941})
$$
f(r)=\int_0^\infty e^{-\alpha r^2}\, d\mu(\alpha).
$$
Then it follows that Theorem~\ref{thm.cohnire} also holds for such $f$ and, in particular, for hypersingular Riesz $s$-potentials $f_s(r)=(r^2)^{-s/2}$  for $s>d$.  
Furthermore, it is shown in \cite{HLSS2017} that the constant $C_{s,d}$ also appears in the context of minimizing the Riesz $s$-energy over infinite point configurations  $\mathcal{C}\subset \mathbb{R}^d$ with a fixed  density $\rho$:
\begin{equation}\label{CsdRd}
C_{s,d}=\inf_{\substack{\text{$\mathcal{C}$ has}\\ \text{density $1$}}}E_{f_s}(\mathcal{C}).
\end{equation}
 Combining \eqref{CsdRd} and Theorem~\ref{thm.cohnire} then provides an alternate proof of Theorem~\ref{thm.csd bound}.  

An outline of the remainder of this article is as follows. In Section 2, we describe the Delsarte-Yudin linear programming lower bounds and the Levenshtein $1/N$-quadrature rules. More thorough treatments can be found in \cite{MEBook}, \cite{Boumova}, and \cite{LevBig}. In Section \ref{section.csdproofs}, we present the proofs of Theorem \ref{thm.csd bound}, Proposition \ref{thm.2tight}, and Theorem \ref{thm.cohnire} using an asymptotic result on Jacobi polynomials from Szeg\H{o} \cite{Szego}. Finally, in Section \ref{section.ulbnum}, we discuss numerically the quality of the bound $A_{s,d}$ and formulate a natural conjecture.

\section{Linear Programming Bounds}

For $\alpha,\beta> -1$, let $\left\{P^{(\alpha,\beta)}_k\right\}_{k=1}^\infty$ denote the sequence of  Jacobi polynomials of respective degrees $k$ that are orthogonal with respect to the weight $\omega^{\alpha,\beta}(t) : = (1-t)^\alpha(1+t)^\beta$ on $[-1,1]$ and normalized by
\begin{equation}
P^{(\alpha,\beta)}_k(1) = 1.
\label{normalization}
\end{equation}While this normalization is crucial for the linear programming methods presented here, we note that many authors choose $P^{(\alpha,\beta)}_k(1) = \binom{k+\alpha}{k}$. For a fixed dimension $d\geq 1$, the Gegenbauer or ultraspherical polynomials are given by $P_k(t) := P^{(\frac{d-2}{2},\frac{d-2}{2})}_k(t)$ with weight $\omega_d(t) := \omega^{(\frac{d-2}{2},\frac{d-2}{2})}(t)$. For our purposes, the so-called \emph{adjacent polynomials}
\begin{equation}
P^{a,b}_k(t): = P^{(\frac{d-2}{2}+a,\frac{d-2}{2}+b)}_k(t), \qquad a,b \in\{0,1\},
\label{eq.adjacent jacobi}
\end{equation}
associated with the weights $\omega_d^{a,b}(t):=(1-t)^a(1+t)^b\omega_d(t),$ play an essential role.

 For functions $f:[-1,1]\rightarrow \mathbb{R}$ that are square integrable with
 respect to $\omega_d$ on $[-1,1]$, we consider its Gegenbauer expansion: $f(t) = \sum_{i=0}^\infty f_kP_k(t),$ where the $f_k$'s are given by

\begin{equation}
f_k := \frac{\int_{-1}^1f(t)P_k(t)\omega_d(t)dt}{\int_{-1}^1 [P_k(t)]^2\omega_d(t)dt}.
\label{eq.gegcoeffs}
\end{equation}
The following result forms the basis for the linear programming bounds for packing and energy on the sphere (see, for example, \cite{Delsarte} or \cite[Theorem 5.3.2]{MEBook}):
\begin{theorem}
If $f:[-1,1]\rightarrow \mathbb{R}$ is of the form

\[f(t) = \sum_{k=0}^{\infty}f_kP_k(t)\]
with $f_k\geq 0$ for all $k\geq 1$ and $\sum_{k=0}^\infty f_k < \infty$, then for any $N$-point subset $\omega_N = \left\{x_1,\ldots,x_N\right\}\subset\mathbb{S}^d$,
\begin{equation}
    \sum_{1\leq i\neq j\leq N} f(\langle x_i, x_j\rangle) \geq f_0 N^2 - f(1)N.
\end{equation}
Moreover, if $h: [-1,1] \rightarrow [0,\infty]$ and $h(t)\geq f(t)$ on $[-1,1]$, then for the energy kernel $K(x,y) := h(\langle x, y\rangle),$
\begin{equation}
E_K(\omega_N)\geq \mathcal{E}_K(\mathbb{S}^d,N)\geq f_0N^2-f(1)N.
\label{eq.yudin energy bound}
\end{equation}
Equality holds in \eqref{eq.yudin energy bound} and $\omega_N$ is an optimal (minimizing) $h$-energy configuration if and only if\

{\rm(i)} h(t)=f(t) for all $t\in\left\{\langle x_i,x_j \rangle: i\neq j\right\}$ and\

{\rm(ii)} \,for all $k\geq 1$, either $f_k = 0$ or $\displaystyle\sum_{i,j = 1}^N P_k(\langle x_i,x_j\rangle) = 0$.
\label{thm.del}
\end{theorem}\

An $N$-point configuration $\omega_N = \left\{x_i\right\}_{i=1}^N\subset \mathbb{S}^d$ is called a \textit{spherical $\tau$-design} if

\[\int_{S^{d}} f(x)d\sigma_d(x) = \frac{1}{N}\sum_{i=1}^N f(x_i)\]
holds for all spherical polynomials $f$ of degree at most $\tau$,
where $\sigma_d$ denotes the normalized surface area measure on $\mathbb{S}^d$.
Using Theorem \ref{thm.del}, Delsarte, Goethals, and Seidel \cite{DGS} obtained an estimate for the minimum
 number of points on $\mathbb{S}^d$ that are necessary for a $\tau$-design. Namely, setting
\[B(d,\tau):=\min \left\{N :  \exists\ \omega_N \subset \mathbb{S}^d\textup{ a spherical $\tau$-design} \right\},\]
they show

\[B(d,\tau)\geq D(d,\tau):=\left\{\arraycolsep=1.4pt\def\arraystretch{1.4}
\begin{array}{ll}
    2\binom{d+k-1}{d} & \textup{ if } \tau = 2k-1,\\

    \binom{d+k}{d} + \binom{d+k-1}{d} & \textup{ if } \tau = 2k.
\end{array}\right.
\]

\begin{definition} A sequence of ordered pairs $\left\{(\alpha_i,\rho_i)\right\}_{i=1}^k$ is said to be a $1/N$-\textit{quadrature rule exact on a subspace} $\Lambda\subset C([-1,1])$ if $1>\alpha_1>\cdots>\alpha_k\geq -1$, $\rho_i>0$ for $i = 1,\ldots,k$, and for all $f\in \Lambda$,
\begin{equation}
    f_0=\frac{1}{\lambda_d}\int_{-1}^1f(t)\omega_d(t)dt = \frac{f(1)}{N}+\sum_{i=1}^k\rho_if(\alpha_i).
    \label{Levquad}
\end{equation}
\end{definition}
Theorem \ref{thm.del} gives rise immediately to the following:

\begin{theorem}
Let $\left\{(\alpha_i,\rho_i)\right\}_{i=1}^k$ be a $1/N$-quadrature rule exact on a subspace $\Lambda$. For $K(x,y) := h(\langle x,y\rangle)$, let $\mathcal{A}_h$ be the set of functions $f$ with $f(t)\leq h(t)$ on $[-1,1]$ that  satisfy the hypotheses of Theorem \ref{thm.del}. Then

\[\mathcal{E}_K(\mathbb{S}^d,N) \geq N^2\sum_{i=1}^k\rho_if(\alpha_i)\]
and
\[\sup_{f\in\Lambda\cap\mathcal{A}_h} N^2\sum_{i=1}^k\rho_if(\alpha_i)\leq N^2\sum_{i=1}^k\rho_ih(\alpha_i).\]
\label{thm.Yudinquad}
\end{theorem}

Levenshtein derives a $1/N$-quadrature given in Theorem \ref{thm.nodeformula} below to obtain the following bound for the maximal cardinality of a configuration $\omega_N\subset\mathbb{S}^d$ with largest inner product $s$. Let

\[A(d,s) := \max\left\{N: \exists\  \omega_N \subset\mathbb{S}^d \,{\rm{with}}\, \langle x_i,x_j\rangle \leq s,\, i\neq j\right\}.\]
Letting $\gamma_k^{a,b}$ denote the greatest zero of $P_k^{a,b}$, we partition $[-1,1]$ into the following disjoint union of countably many intervals. For $\tau = 1,2,\ldots$,
\[I_\tau := \left\{\arraycolsep=5pt\def\arraystretch{1.8}
\begin{array}{ll}
  [\gamma_{k-1}^{1,1},\gamma_k^{1,0}] & \textup{ if } \tau = 2k-1,\\

   [\gamma_k^{1,0},\gamma_k^{1,1}] & \textup{ if } \tau = 2k,
\end{array} \right.
\]
which are well defined by the interlacing properties $\gamma_{k-1}^{1,1}<\gamma_k^{1,0}<\gamma_k^{1,1}$. Then
\[A(d,s)\leq L(d,s),\]
where
\begin{equation}L(d,s) := \left\{\arraycolsep=5pt\def\arraystretch{1.8}
\begin{array}{ll}
   L_{2k-1}(d,s) := \binom{k+d-2}{k-1}[\frac{2k+d-2}{d} - \frac{P_{k-1}(s)-P_k(s)}{(1-s)P_k(s)}] & \textup{ if } s\in I_{2k-1},\\

   L_{2k}(d,s): = \binom{k+d-1}{k}[\frac{2k+d}{d}-\frac{(1+s)(P_k(s)-P_{k+1}(s))}{(1-s)(P_k(s)+P_{k+1}(s))}] & \textup{ if } s\in I_{2k}.
\end{array}\right.
\label{eq.Lev}
\end{equation}
The function $L(d,s)$ is called the \emph{Levenshtein function}. For fixed $d$, it is continuous and increasing in $s$ on $[-1,1]$. The formula for the Levenshtein function is such that the quadrature nodes given in Theorem \ref{thm.nodeformula} below will have weight $1/N$ at the node $\alpha_0 = 1$. At the endpoints of the intervals $\mathcal{I}_\tau$,
\begin{equation}
\begin{split}
    L_{2k-2}(d,\gamma^{1,1}_{k-1}) & = L_{2k-1}(d,\gamma^{1,1}_{k-1}) = D(d,2k-1),\\
    L_{2k-1}(d,\gamma^{1,0}_{k-1}) & = L_{2k}(d,\gamma^{1,0}_{k-1}) = D(d,2k),
\end{split}
\label{Levendpts}
\end{equation}
where $L_{\tau}$ denotes the restriction of $L$ to the interval $I_{\tau}.$\

Setting
$$
r^{a,b}_i: = \left(\frac{1}{\lambda_d^{a,b}}\int_{-1}^1(P^{a,b}_i(t))^2\omega^{a,b}_{d}(t)\,dt\right)^{-1},
$$
where
$\lambda_d^{a,b}:=\int_{-1}^1\omega^{a,b}_{d}(t)\,dt,$
we  define
\begin{equation}
Q^{a,b}_k(x,y):= \sum_{i=0}^k r^{a,b}_i P^{a,b}_i(x)P^{a,b}_i(y).
\label{eq.sum}
\end{equation}
By the Christoffel Darboux formula (see \cite[Section 3.2]{Szego}),

\begin{align}
\label{CD1}Q^{a,b}_k(x,y) & =r^{a,b}_km^{a,b}_k\bigg(\frac{P^{a,b}_{k+1}(x)P^{a,b}_k(y)-P^{a,b}_k(x)P^{a,b}_{k+1}(y)}{x-y}\bigg),\ \ \ \ \ \ \ \ x\neq y\\
\label{CD2}Q^{a,b}_k(x,x) & =r^{a,b}_km^{a,b}_k\left[(P^{a,b}_{k+1})^{'}(x)P^{a,b}_k(x)-(P^{a,b}_k)^{'}(x)P^{a,b}_{k+1}(x)\right],
\end{align}
where $m^{a,b}_i:=l^{a,b}_i/l^{a,b}_{i+1}$ and  $l^{a,b}_i$ is the leading coefficient of $P^{a,b}_i$.\

The following $1/N$-quadrature rule proven in \cite[Theorems 4.1 and 4.2]{Lev} plays an essential role in establishing Theorem \ref{thm.csd bound}.

\begin{theorem}
For $N\in\mathbb{N}$, let $\tau$ be such that $N \in (D(d,\tau),D(d,\tau+1)]$, and let $\alpha_1 = \beta_1 = s$ be the unique solution to

\[N = L(d,s).\]

{\rm(i)} If $\tau  =2k-1$, define nodes $1>\alpha_1>\cdots>\alpha_k> -1$ as the solutions of
\begin{equation}
(t-s)Q_{k-1}^{1,0}(t,s) = 0
\label{nodes}
\end{equation}
with associated weights
\begin{equation}
\rho_i = \frac{\lambda_d^{1,0}}{\lambda_d(1-\alpha_{i})Q^{1,0}_{k-1}(\alpha_i,\alpha_i)}.
    \label{weights}
\end{equation}
Then $\left\{(\alpha_i,\rho_i)\right\}_{i=1}^k$ is a $1/N$-quadrature rule exact on $\Pi_{2k-1}$.

\

{\rm(ii)} If $\tau =2k$, define nodes $1>\beta_1>\cdots>\beta_{k+1} = -1$ as the solutions of
\begin{equation}
(1+t)(t-s)Q_{k-1}^{1,1}(t,s) = 0
\end{equation}
with associated weights
\begin{equation}
\begin{split}
\eta_i & = \frac{\lambda_d^{1,1}}{(1-\beta_i^2)Q^{1,1}_k(\beta_i,\beta_i)}, \ \ \ \ \ \ \ \ i = 1,\ldots,k,\\
\eta_{k+1} & = \frac{Q_k(s,1)}{Q_k(-1,-1)Q_k(s,1)-Q_k(-1,1)Q_k(s,-1)}.
\end{split}
\end{equation}
Then $\left\{(\beta_i,\eta_i)\right\}_{i=1}^{k+1}$ is a $1/N$-quadrature rule exact on $\Pi_{2k}$.
\label{thm.nodeformula}
\end{theorem}
Here and below $\Pi_{m}$ denotes the collection of all algebraic polynomials of degree
at most $m$.
\begin{remark}
At the endpoints we also have that for $N = D(d,2k)$, $\left\{(\alpha_i,\rho_i)\right\}_{i=1}^k$ is exact on $\Pi_{2k}$ and for $N = D(d,2k+1)$,  $\left\{(\beta_i,\eta_i)\right\}_{i=1}^{k+1}$ is exact on $\Pi_{2k+1}$.
\end{remark}
 The above quadrature rules were used by Boyvalenkov et. al to derive the following universal lower bounds for the energy of spherical configurations.

\begin{theorem}{\rm(\cite{Petersquared})}
Let $N$ be fixed and $h(t)$ denote an absolutely monotone potential on $[-1,1)$. Suppose $\tau = \tau(d,N)$ is such that $N\in (D(d,\tau),D(d,\tau+1)]$ and let $k = \lceil{\frac{\tau+1}{2}}\rceil$. If $\{(\alpha_i,\rho_i)\}_{i=1}^k$ is the $1/N$-quadrature rule of Theorem \ref{thm.nodeformula}, then

\begin{equation}
\mathcal{E}_h(\mathbb{S}^d,N)\geq N^2\sum_{i=1}^k\rho_ih(\alpha_i).
\label{PeterULB}
\end{equation}
\label{thm.hbound}
\end{theorem}
An analogous statement holds for the pairs $(\beta_i,\eta_i)$ of Theorem \ref{thm.nodeformula}(ii), but
we shall not make use of it in our proofs.\

Taking into account Theorem \ref{thm.Yudinquad}, inequality \eqref{PeterULB} provides an optimal linear programming lower bound for the subspace $\Lambda = \Pi_k$. As an application, we now show that Theorem \ref{thm.hbound} recovers the first-order asymptotics for integrable potentials.

\begin{example} If $h(t)$ is any absolutely monotone function that is also integrable with respect to $\omega_d(t)$ on $[-1,1]$, then
\begin{equation}
    \lim_{N\to\infty} \frac{\mathcal{E}_h(\mathbb{S}^d,N)}{N^2} \geq \frac{1}{\lambda_d}\int_{-1}^1 h(t)\omega_d(t)\,dt,
    \label{intasymp}
\end{equation}
where $\lambda_d$ is defined in \eqref{eq.lambdad}.
\end{example}
\begin{remark}It is a classical result of potential theory that the limit exists and equality holds in (\ref{intasymp}); see \cite{Land}.
\end{remark}

\noindent \emph{Proof of \eqref{intasymp}.}
First suppose $h(t)$ is continuous on $[-1,1]$. For $\epsilon > 0$, let $f(t)$ be a polynomial of degree $\leq 2k-1$ such that $|f(t)-h(t)| \leq \epsilon$ uniformly on $[-1,1]$. Setting $(\alpha_0,\rho_0):= (1,1/N)$, we note that the weights $\rho_i$ given in (\ref{weights}) are positive for $i= 0,\ldots, k$ and that $\sum_{i=0}^k\rho_i = 1$. From (\ref{Levquad}), we have with $\alpha_i=\alpha_i(N),\,\rho_i=\rho_i(N), \,k=k(N),$
\begin{align*}
    \bigg\vert\frac{1}{\lambda_d}\int_{-1}^1h(t)\omega_d(t)\,dt&-\sum_{i=0}^k\rho_ih(\alpha_i)\bigg|\\
    & \leq \frac{1}{ \lambda_d}\int_{-1}^1|h(t)-f(t)|\omega_d(t)\,dt + \sum_{i=0}^k\rho_i|f(\alpha_i)-h(\alpha_i)|\\
    & \leq 2\epsilon \to 0 \,\,\,{\rm{as}}\,\, N\to\infty.
\end{align*}
Since $\rho_0h(\alpha_0)=h(1)/N \to 0$ as $N\to\infty$, inequality (\ref{intasymp}) follows.\

 Next suppose $h(t)$ is integrable and $g_m\nearrow h$ a sequence of continuous functions increasing to $h$ (for existence, consider $g_m(t):=h((1-1/m)(t+1)-1)$). By the Monotone Convergence Theorem and a similar string of inequalities as above, it follows that
\[\lim_{k\to\infty}\sum_{i=1}^k\rho_ih(\alpha_i) = \frac{1}{ \lambda_d}\int_{-1}^1 h(t)\omega_d(t),\]
which concludes the proof.\\

We remark that another feature of Theorem \ref{thm.nodeformula} is that it includes a best-packing result of Levenshtein \cite{LevPacking},\cite{LevBig}, which asserts the following: if $\omega_N=\left\{x_1,\ldots, x_N\right\}$
is any $N$-point configuration on $\mathbb{S}^d$ and $\delta(\omega_N):=\max_{i\neq j}\langle x_i,x_j\rangle$, then
\begin{equation}
\delta(\omega_N)\geq \alpha_1,
\label{eq.spheresep}
\end{equation}
where $\alpha_1=\alpha_1(N)$ is as given in Theorem \ref{thm.nodeformula}. This follows by considering
absolutely monotone approximations to the potential

%\begin{corollary}
%Let $\omega_N^*=\left\{x_1,\ldots x_N\right\}$ be an optimal packing configuration on $\mathbb{S}^d$ and $y_N:=\max_{i\neq j}\langle x_i,x_j\rangle$. Then
%\begin{equation}
%y_N\geq \alpha_1
%\label{eq.spheresep}
%\end{equation}
%
%\label{spheresep}
%\end{corollary}
%\begin{proof}
%Let
\[
h(t) = \left\{
\begin{array}{ll}
   \infty & \textup{ if } t\geq\alpha_1\\

    0 & \textup{ if } t<\alpha_1.
\end{array}\right.
\] Indeed, if  $\delta(\omega_N)<\alpha_1$, then $E_{h}(\omega_N)=0$, but $\sum_{i=1}^k\rho_ih(\alpha_i)=\infty$, contradicting \eqref{PeterULB}.

%Then $E_{h_\alpha} (\omega_N^*) \geq N^2\sum_{i=1}^k\rho_ih_\alpha(\alpha_i)$. If $y_N<\alpha_1$, then $E_{h_{y_N}} (\omega_N^*)=0$, but $\sum_{i=1}^k\rho_ih_{y_N}(\alpha_i)=\infty$.

\section{Proofs of Theorems \ref{thm.csd bound}, \ref{thm.cohnire}, and Proposition~\ref{thm.2tight}}\label{section.csdproofs}

	Our approach will be to find the asymptotic expansion of the right-hand side of (\ref{PeterULB}) as $N\to \infty$. Throughout this section we assume that  $\alpha,\beta>-1$.  We will make use of the following result from Szeg\H{o} (see \cite[Theorem 8.1.1]{Szego}) adjusted by normalization (\ref{normalization}):

 \begin{theorem} Locally uniformly in the complex $z$-plane,
	\[\lim_{k\to\infty}P^{(\alpha,\beta)}_k\bigg(\cos\frac{z}{k}\bigg) = \lim_{k\to\infty}P^{(\alpha,\beta)}_k\bigg(1-\frac{z^2}{2k^2}\bigg) = \Gamma(\alpha+1)\bigg(\frac{z}{2}\bigg)^{-\alpha}J_{\alpha}(z).\]
	\label{thm.szego}
\end{theorem}
This gives the following immediate corollary:

\begin{corollary}

If $-1< \gamma_{k,k}<\dots<\gamma_{k,1}<1$ are the zeros of $P^{(\alpha,\beta)}_k$ and $z_i$ is the $i$-th smallest positive zero of the Bessel function $J_\alpha(z)$, then
\begin{equation}
\lim_{k\to \infty} k\cos^{-1}(\gamma_{k,i}) = z_{i}.
\label{eq.jacobi zeros}
\end{equation}
\label{thm.jacobi zeros}
\end{corollary}

Recalling definition \eqref{eq.adjacent jacobi} and making use of well-known properties of the derivatives, norms, and leading coefficients of the Jacobi polynomials (see, e.g., \cite[Chapter 4]{Szego})
%, using $\alpha = d/2$, $\beta = (d-2)/2$ and adjusting for normalization (\ref{normalization})) we obtain the following asymptotic formulas
we obtain the following asymptotic formulas as $k\to\infty$:
\begin{equation}
	\begin{split}
	    \frac{\textup{d}}{\textup{d}t}P^{1,0}_k(t) &= \frac{1}{2}(k+d)\frac{\binom{k+\frac{d+2}{2}}{k}}{\binom{k+\frac{d}{2}}{k}}P^{2,1}_{k-1}(t)\\
	    %&=\frac{1}{2}(k+d)\bigg(\frac{2k+d+2}{d+2}\bigg)P^{2,1}_{k-1}(t)\\	
	    &=\bigg(\frac{k^2}{d+2}+o(k^2)\bigg)P^{2,1}_{k-1}(t).
	    \label{eq.jacobi derivative}
	\end{split}
	\end{equation}
Furthermore,
	
	\begin{equation}
	\begin{split}
	   \frac{r^{1,0}_k}{\lambda_d^{1,0}} & = \bigg(\int_{-1}^1(P^{1,0}_k(t))^2\omega^{1,0}(t)\,dt\bigg)^{-1}\\
	    & = \Bigg(\frac{2^d\Gamma(k+\frac{d+2}{2})\Gamma(k+\frac{d}{2})}{\binom{k+\frac{d}{2}}{k}^2(2k+d)\Gamma(k+d)\Gamma(k+1)}\Bigg)^{-1}\\
	    %& = \frac{(2k+d)^2\Gamma(k+d)}{2^(d+1)\Gamma(\frac{d+2}{2})^2\Gamma(k+1)}\\
	    & = \frac{k^{d+1}}{2^{d-1}\Gamma(\frac{d+2}{2})^2}+o(k^{d+1}).
	\end{split}
	\label{eq.jacnorms}
	\end{equation}
Lastly, recalling that  $l^{1,0}_k$ is the leading coefficient of $P^{1,0}_k(t)$,
\[l^{1,0}_k = \frac{\Gamma(2k+d)}{\binom{k+\frac{d}{2}}{k}2^k\Gamma(k+1)\Gamma(k+d)},\]
which yields for the  ratio

	\begin{equation}
	\begin{split}
	    m^{1,0}_k  = \frac{l^{1,0}_k}{l^{1,0}_{k+1}} & = \bigg(\frac{2(k+1)(k+d)}{(2k+d+1)(2k+d)}\bigg)\bigg(\frac{2k+2+d}{2k+2}\bigg)\\
	    & = \frac{1}{2}+o(1).
	\end{split}
	\label{eq.coeff ratio}
	\end{equation}
\begin{remark}Generalizing equations (\ref{eq.jacobi derivative}) - (\ref{eq.coeff ratio}) to $P^{(\alpha,\beta)}_k(t)$ we obtain
\begin{equation}
\frac{\textup{d}}{\textup{d}t}P^{(\alpha,\beta)}_k(t) = \bigg(\frac{k^2}{2(\alpha+1)}+o(k^2)\bigg)P^{(\alpha+1,\beta+1)}_{k-1}(t),
\label{eq.jac deriv general}
\end{equation}
\begin{equation}
r^{(\alpha, \beta)}_k = O(k^{2\alpha+1}),\ \textup{and}\ \ \ \ \ \ \ \ \ \ \ \ \ \ \ \ \ \ \ \ \ \ \ \ \ \ \ \ \ \ \ \
\end{equation}
\begin{equation}
m^{(\alpha,\beta)}_k = \frac{1}{2}+o(1).\ \ \ \ \ \ \ \ \ \ \ \ \ \ \ \ \ \ \ \ \ \ \ \ \ \ \ \ \ \ \ \ \ \ \ \ \ \ \
\end{equation}
\end{remark}
	We also need the following additional lemmas.

\begin{lemma}
 Let $p_k(t) := P^{(\alpha,\beta)}_k(t)$ be a sequence of Jacobi polynomials. If $z\in \mathbb{R}$ is fixed such that $\displaystyle\lim_{k\to\infty}p_k(\cos\frac{z}{k}) = c $ and $\beta_k,\,-1\leq \beta_k\leq 1,$ is a sequence satisfying
\begin{equation}
\lim_{k\to\infty}k\cos^{-1}(\beta_k) = z,
\label{betalim}
\end{equation}
then
\begin{equation}
\lim_{k\to\infty}p_{k+ j}(\beta_k) = c,
\label{lem.eq.subindex}
\end{equation}
for any fixed $j\in \mathbb{Z}$.
\label{lem.subindex}
\end{lemma}
\begin{proof}

First, since $\displaystyle\lim_{k\to\infty}(k+ j)\cos^{-1}(\beta_k) = \lim_{k\to\infty}k\cos^{-1}(\beta_k)$, by making the substitution $k = k+ j$ it suffices to establish equation (\ref{lem.eq.subindex}) for the case $j=0$. From (\ref{betalim}), we have that
\[\epsilon_k:= |\beta_k-\cos\frac{z}{k}| = o\bigg(\frac{1}{k^2}\bigg).\]
 Applying the mean value theorem, equation (\ref{eq.jac deriv general}), and using the fact that $p_k$ is uniformly bounded in $k$ on $[-1,1]$ (see e.g. \cite{ErdMagNev}) we get with $p_{k-1}^{1,1}:=P_{k-1}^{\alpha+1\,\beta+1},$
\[|p_k(\beta_k)-p_k(\cos\frac{z}{k})|  = p_k'(\xi_k)\epsilon_k = k^2\tilde{c}p_{k-1}^{1,1}(\xi_k)\epsilon_k = o(1),\]
for some $\xi_k$ between $\beta_k$ and $\cos(\frac{z}{k})$, and $\tilde{c}>0$.
\end{proof}

A stronger version of Lemma \ref{lem.subindex} holds when $c=0$.

\begin{lemma}  Let $-1< \gamma_{k,k}<\dots<\gamma_{k,1}< 1$ be the zeros of $p_k(t):= P^{(\alpha,\beta)}_k(t)$, and denote by $z_{i}$  the $i$-th smallest positive zero of the Bessel function $J_\alpha(z)$. Then for all $i = 1,2,\ldots$,

\[\lim_{k\to\infty}kp_{k-1}(\gamma_{k,i}) = 2\Gamma(\alpha+1)\bigg(\frac{z_{i}}{2}\bigg)^{-\alpha+1}J_{\alpha+1}(z_{i}).\]

\label{lem.aha}
\end{lemma}

\begin{proof}
	By Corollary \ref{thm.jacobi zeros},
	\[\gamma_{k,i} = 1-\frac{z_{i}^2}{2k^2} + o\bigg(\frac{1}{k^2}\bigg),\]
	which implies
	\[\delta_k:=|\gamma_{k,i}-\gamma_{k-1,i}| = \frac{z_{i}^2}{k^3}+o\bigg(\frac{1}{k^3}\bigg),\,\,\mathrm{as}\,\, k\to \infty.\]
	By the interlacing properties of the zeros of Jacobi polynomials, we see that $\gamma_{k,i}>\gamma_{k-1,i}$ and we can drop the absolute value in $\delta_k$. Expanding the Taylor series for $p_{k-1}(t)$ around the zero $\gamma_{k-1,i}$,  we have
	
	\[kp_{k-1}(\gamma_{k,i}) = k\delta_kp'_{k-1}(\gamma_{k-1,i}) + \frac{k\delta_k^2p^{''}_{k-1}(\gamma_{k-1,i})}{2}+\cdots\]
	
	Each successive derivative term beyond the first has order $o(1)$ since by repeated application of \eqref{eq.jac deriv general} and Lemma \ref{lem.subindex}, $p_k^{(j)}(t) = O(k^{2j})p^{j,j}_{k-j}(t) = O(k^{2j})$ while on the other hand $\delta_k^j = O(1/k^{3j})$. Thus,
	\[kp_{k-1}(\gamma_{k,i}) = \frac{z^2_{i}}{2(\alpha+1)}p^{1,1}_{k-2}(\gamma_{k-1,i})+o(1)\,\, \mathrm{as} \,\, k\to\infty.\]
	Now by Theorem \ref{thm.szego} and Lemma \ref{lem.subindex}, we obtain the result.
\end{proof}
We are now ready to prove the main theorem.
\begin{proof}[\textbf{Proof of Theorem \ref{thm.csd bound}}]

In the case of Riesz energy, we have
\[K_s(x,y) = h_s(\langle x,y\rangle) = (2-2\langle x,y\rangle)^{-s/2}\]
We consider the subsequence
\begin{equation}
N_k := D(d,2k) =  \binom{d+k}{d} + \binom{d+k-1}{d} = \frac{2}{\Gamma(d+1)}k^{d} + o(k^d).
\label{eq.subsequence}
\end{equation}
By Theorem \ref{poppy seed} it suffices to prove

\begin{equation}\lim_{k\to \infty}\frac{\mathcal{E}_s(N_k,\mathbb{S}^d)}{N_k^{1+s/d}} \geq \frac{A_{s,d}}{\mathcal{H}_d(\mathbb{S}^d)^{s/d}}
\label{eq.new bound}
\end{equation}
where \[\mathcal{H}_d(\mathbb{S}^d) = \frac{2\pi^{\frac{d+1}{2}}}{\Gamma\big(\frac{d+1}{2}\big)}. \]
Along the subsequence $N_k$, from (\ref{Levendpts}), $\alpha_1 = \gamma^{1,0}_{k,1},$ where $\gamma^{1,0}_{k,i}$ is the $i$-th largest zero of $P^{1,0}_k(t)$ and
\begin{align*}(t-\alpha_1)Q^{1,0}_{k-1}(t,\alpha_1) & = r_{k-1}m_{k-1}(P^{1,0}_{k}(t)P^{1,0}_{k-1}(\alpha_1)-P^{1,0}_{k-1}(t)P^{1,0}_k(\alpha_1))\\
& = r_{k-1}m_{k-1}(P^{1,0}_{k}(t)P^{1,0}_{k-1}(\alpha_1));\\
\end{align*}
thus the quadrature nodes are given by
\[\alpha_i = \gamma^{1,0}_{k,i}, \ \ \ \ \ \ \ \ i = 1,2,\ldots, k.\]
For a fixed $m$ and all $k\geq m$ we have by Theorem \ref{thm.hbound}
\[\frac{\mathcal{E}_s(N_k)}{N_k^{1+s/d}}\geq\frac{\sum_{i=1}^k\rho_ih_s(\alpha_i)}{N_k^{-1+s/d}}\geq\frac{\sum_{i=1}^m\rho_ih_s(\alpha_i)}{N_k^{-1+s/d}}.\]\

For a fixed $i\leq m$, we next establish asymptotics for $\rho_ih(\alpha_i)$. By Corollary \ref{thm.jacobi zeros} we have

\begin{equation}
\lim_{k\to \infty} \frac{h_s(\alpha_i)}{k^{s}} = \lim_{k\to \infty} \frac{(2-2\alpha_i)^{-s/2} }{k^s}= (z_{i})^{-s},
\label{eq.halpha}
\end{equation}
and by (\ref{eq.jacobi derivative}) and Lemma \ref{lem.subindex},
\begin{equation}
\lim_{k\to\infty} \frac{(P^{1,0}_k)^{'}(\alpha_i)}{k^2} = \frac{\Gamma(d/2+2)}{d+2}\bigg(\frac{z_{i}}{2}\bigg)^{-\frac{d+2}{2}}J_{d/2+1}(z_{i}).
\label{eq.derivative limit}
\end{equation}
Furthermore, from Lemma \ref{lem.aha}, it follows that

\begin{equation}
\lim_{k\to\infty}kP^{1,0}_{k-1}(\alpha_i) = 2\Gamma(d/2+1)\bigg(\frac{z_{i}}{2}\bigg)^{-\frac{d-2}{2}}J_{d/2+1}(z_{i}).
\label{eq.aha}
\end{equation}
From the weight formula given in equation (\ref{weights}) and the Cristoffel-Darboux formula (\ref{CD2}) we deduce that

\begin{equation}
\lim_{k\to\infty}k^d\rho_i = \lim_{k\to\infty}k^d\left(\frac{\lambda_d}{\lambda_d^{1,0}}(1-\alpha_{i})r^{1,0}_{k-1}m^{1,0}_{k-1}(P^{1,0}_{k})^{'}(\alpha_i)P^{1,0}_{k-1}(\alpha_i)\right)^{-1},
\label{eq.weight limit 1}
\end{equation}
and combining equations (\ref{eq.jacobi zeros}),(\ref{eq.jacnorms}),(\ref{eq.coeff ratio}),(\ref{eq.derivative limit}), and (\ref{eq.aha}), this yields

\begin{equation}
\begin{split}
\lim_{k\to\infty}k^d\rho_i = \ \ \ \ \ \ \ \ \ \ \ \ \ \ \ \ \ \ \ \ \ \ \ \ \ \ & \\
    \bigg[\lambda_d\bigg(\frac{z_{i}^2}{2}\bigg)\bigg(\frac{1}{2^{d-1}\Gamma(d/2+1)^2}\bigg)\frac{1}{2}\bigg(& \frac{\Gamma(d/2+2)}{d+2}\bigg(\frac{z_{i}}{2}\bigg)^{-d/2-1}J_{d/2+1}(z_{i})
\bigg)\\
& \cdot \ 2\Gamma(d/2+1)\bigg(\frac{z_{i}}{2}\bigg)^{-d/2+1}J_{d/2+1}(z_{i}) \bigg]^{-1}.
\end{split}
\label{weight limit 2}
\end{equation}
Simplifying gives
\begin{equation}
\lim_{k\to\infty}k^d\rho_i = \frac{2}{\lambda_d z_{i}^{2-d}\big(J_{d/2+1}(z_{i})\big)^{2}}.
\label{weight limit 3}
\end{equation}
Finally, combining the asymptotics for $N_k$, $h_s(\alpha_i)$, and $\rho_i$, equations (\ref{eq.subsequence}), (\ref{eq.halpha}), and (\ref{weight limit 3}) respectively, we obtain

\begin{equation}
\lim_{k\to\infty}\frac{\rho_ih(\alpha_i)}{N_k^{s/d-1}} = \frac{2}{\lambda_d\big(\frac{2}{\Gamma(d+1)}\big)^{s/d-1}z_{i}^{2-d+s}\big(J_{d/2+1}(z_{i})\big)^{2}},
\label{eq.single term limit}
\end{equation}
and thus
\[\frac{C_{s,d}}{\mathcal{H}_d(\mathbb{S}^d)^{s/d}} = \lim_{k\to\infty}\frac{\mathcal{E}_s(N_k,\mathbb{S}^d)}{N_k^{1+s/d}}\geq \sum_{i=1}^m\frac{2}{\lambda_d\big(\frac{2}{\Gamma(d+1)}\big)^{s/d-1}z_{i}^{2-d+s}\big(J_{d/2+1}(z_{i})\big)^{2}}.\]
Multiplying by $\mathcal{H}_d(\mathbb{S}^d)^{s/d}$ and letting $m\to \infty$ gives (\ref{eq.new bound}) and hence (\ref{eq.Asd}).
\end{proof}

\begin{proof}[\textbf{Proof of Proposition \ref{thm.2tight}}]
We first establish the limit involving $\xi_{s,d}$:
\begin{equation}\label{xilim}
\begin{split}
\lim_{s\to d^+}(s-d)\xi_{s,d}&=\lim_{s\to d^+}d\bigg[\frac{\pi^{d/2}\Gamma(1+\frac{s-d}{2})}{\Gamma(1+\frac{s}{2})}\bigg]^{s/d}
=\frac{d\pi^{d/2}}{\Gamma(1+\frac{d}{2})}=\frac{2\pi^{d/2}}{\Gamma(\frac{d}{2})}.
\end{split}
\end{equation}

If $\Lambda$ is a $d$-dimensional lattice with co-volume $|\Lambda|>0$ then it is known (see \cite{Terras}) that the Epstein zeta function has a simple pole at $s=d$ with residue
\begin{equation}
\frac{2\pi^{d/2}}{\Gamma(d/2)|\Lambda |}.\end{equation}
Proposition~\ref{bhsbound}, the bound  \eqref{eq.lattice bound}, and \eqref{xilim} then show
\begin{equation}
\lim_{s\to d^+} (s-d)C_{s,d} = \frac{2\pi^{d/2}}{\Gamma(\frac{d}{2})}.
\end{equation}

Finally, we establish the limit involving $A_{s,d}$.   The well-known asymptotic behavior of $J_{\frac{d}{2}+1}(z)$ \cite{Szego}, as $z\to\infty$, is given by
 \begin{equation}
 J_{\frac{d}{2}+1}(z) = -\sqrt{\frac{2}{\pi z}}\big(\cos\big(z-(d-3)\frac{\pi}{4}\big)+O\big(z^{-3/2}\big)\big)
 \label{asympBessel}
 \end{equation}
 and $z_n$, the $n$-th zero of the $J_{\frac{d}{2}}(z)$, is given by (see \cite{Bessel})
 \begin{equation}
 z_n = n\pi+(d-1)\frac{\pi}{4}+O(n^{-1}).
 \label{asmpyBeszero}
 \end{equation}
Thus,
$$
J_{\frac{d}{2}+1}(z_n)^{-2}=\frac{\pi z_n}{2}+O(n^{-1}),
$$
and so we have
\[\sum_{n=1}^\infty \frac{1}{z_n^{s-d+1}J_{\frac{d}{2}}(z_n)^2} = \frac{\pi}{2}\sum_{n=1}^\infty\frac{1}{z_n^{s-1}+a_n} = \frac{1}{2\pi^{s-d}}\sum_{n=1}^\infty\frac{1}{(n+(d-1)/4+b_n)^{s-1}+a_n}\]
where $a_n$, $b_n  =o(1)$. As $s\to d^+$, this sum approaches the Hurwitz zeta function, $\zeta(s-d+1,(d+3)/4),$ where
\begin{equation}
\zeta(s,q) := \sum_{n=0}^\infty \frac{1}{(n+q)^{s}}.
\label{Hurzeta}
\end{equation}
That is,
\begin{equation}
\lim_{s\to d^+} \frac{\sum_{n=1}^\infty ((n+(d-1)/4+b_n)^{s-d+1}+a_n)^{-1}}{\zeta(s-d+1,(d+3)/4)} = 1.
\label{eq.Hurlimit}
\end{equation}
Indeed, suppose $a = \sup |a_n|$ and $b = \sup |b_n|$. Then,
\begin{align*}
    &\sum_{n=1}^\infty\frac{1}{(n+(d-1)/4+b_n)^{s-d+1}+a_n} \geq\sum_{n=1}^\infty\frac{1}{(n+(d-1)/4+b)^{s-d+1}+a}\\ &\geq\sum_{n=1}^\infty\frac{1}{(n+(d-1)/4+b+a)^{s-d+1}}  = \sum_{n=0}^\infty\frac{1}{(n+(d+3)/4+b_n)^{s-d+1}+a_n} \\& = \zeta(s-d+1,(d+3)/4+a+b),
    \end{align*}
and similarly
\[\sum_{n=1}^\infty\frac{1}{(n+(d-1)/4+b_n)^{s-d+1}+a_n}\leq \zeta(s-d+1,(d+3)/4-a-b).\]
Since $\zeta(s,q)\to \infty$ as $s\to 1^+$ (and the terms in the series in \eqref{eq.Hurlimit} stay bounded) the limit \eqref{eq.Hurlimit} holds.
In fact $\zeta(s,q)$ has a simple pole of residue 1 at $s=1$ for all $q$ and so we obtain:
\begin{equation*}
\begin{split}
\lim_{s\to d^+} (s-d)A_{s,d} &= \lim_{s\to d^+}\bigg[\frac{\pi^{\frac{d+1}{2}}\Gamma(d+1)}{\Gamma(\frac{d+1}{2})}\bigg]^{s/d}\frac{4(s-d)}{\lambda_d\Gamma(d+1)}\sum_{i=1}^\infty(z_{i})^{d-s-2}\big(J_{d/2+1}(z_{i})\big)^{-2}\\
&=\lim_{s\to d^+}\frac{4\pi^{d/2}}{\Gamma(\frac{d}{2})}\frac{(s-d)}{2}\zeta(s-d+1,(d+3)/4)=\frac{2\pi^{d/2}}{\Gamma(\frac{d}{2})},
\end{split}
\end{equation*}
which completes the proof of Proposition~\ref{thm.2tight}.
\end{proof}

\begin{proof}[\textbf{Proof of Theorem \ref{thm.cohnire}}]
For a fixed $\rho$ and a Gaussian potential $f(|x-y|) = h(\langle x,y\rangle) = e^{-\alpha(2-2\langle x,y\rangle)}$, set
$$c:= (a_d \rho)^{1/d},\,\,\, \mathrm{where} \,\,\,
a_d:=\frac{(d+1)\pi^{\frac{d+1}{2}}}{\Gamma(1+\frac{d+1}{2})}=\frac{2\pi^{\frac{d+1}{2}}}{\Gamma(\frac{d+1}{2})}$$
is the area of $\mathbb{S}^d$, and let
\[f_N(|x-y|) = h_N(\langle x,y\rangle) := e^{-\alpha\frac{2-2\langle x,y\rangle}{(cN^{-1/d})^2}}.\]
Our approach is to first obtain estimates for the $h_N$-energy  of $N$-point configurations on the sphere $\mathbb{S}^d$.\

For each $N$, $h_N$ is absolutely monotone on $[-1,1)$, and so Theorem \ref{thm.hbound} holds.
We apply the same asymptotic argument as in the proof of Theorem \ref{thm.csd bound} to $h_N(t)$. In particular we sample along the subsequence
\[N_k := D(d,2k),\]
where the nodes $\alpha_i$ are given by the zeros of $P_k^{1,0}(t)$. Using the asymptotic formulas for $N_k$, the quadrature nodes $\alpha_i$, and the weights $\rho_i$, we obtain from Corollary \ref{thm.jacobi zeros} and \eqref{weight limit 3} that
\begin{equation}
\liminf_{N\to \infty} \frac{\mathcal{E}_{h_{N}}(\mathbb{S}^d,N)}{N}\geq
\frac{4}{\lambda_d\Gamma(d+1)}\sum_{i=1}^\infty\frac{z_i^{d-2}}{(J_{d/2+1}(z_i))^2}e^{-\alpha\big(\frac{z_i}{c(2/\Gamma(d+1))^{-1/d}}\big)^2}.
\label{GaussSphasymp}
\end{equation}

Let $0<\epsilon<1$.  Then there is a collection  $\{ C(a_\ell,r_\ell)\colon \ell=1, 2, \ldots, L\}$ of disjoint closed spherical caps on $\mathbb{S}^d$ such that $r_\ell<\epsilon$ and $$\sum_{\ell=1}^L\mathcal{H}_d( C(a_\ell,r_\ell))\ge (1-\epsilon)\mathcal{H}_d(\mathbb{S}^d).$$ Using \eqref{HdCap} and the fact that the caps are disjoint, it follows that there is a constant
$\kappa_1>0$, independent of $\epsilon$, such that
\begin{equation}\label{kappa1}
 (1+\kappa_1 \epsilon)^{-1}d\lambda_d\le \sum_{\ell=1}^L r_\ell^d \le d\lambda_d(1+\kappa_1\epsilon).
\end{equation}
Furthermore, there  are mappings $\phi_\ell:B^d(r_\ell)\to C(a_\ell,r_\ell)$, $\ell=1,2, \ldots, L$ and a constant $\kappa_2$ (again independent of $\epsilon$)  such that
\begin{equation}\label{phiineq}
|\phi_\ell(x)-\phi_\ell(y)|\ge (1- \kappa_2\epsilon)|x-y|, \qquad (x,y\in B^d(r_\ell)).
\end{equation}

Let $\mathcal{C}$ be a configuration in $\mathbb{R}^d$ with density $\rho$ and $f_\alpha$-energy $E_{f_\alpha}(\mathcal{C})$; i.e., the limits in Definitions~\ref{lowerfenergy} and \ref{lowerdensity} both exist.  Then, as $R\to \infty$,
we have for any $\alpha>0$,
\begin{equation}\label{cardCr}
   \#(\mathcal{C}\cap B^d(R))={\rho\,  \textup{vol}(B^d(R))}(1+o(1)),
\end{equation}
and
\begin{equation}
E_{f_\alpha}\left(\mathcal{C}\cap B^d(R)\right)\le [\rho\,  \textup{vol}(B^d(R))]\, E_{f_\alpha}(\mathcal{C})(1+o(1)).
\end{equation}
For $\ell=1, 2, \ldots, L$, let
$$
\omega_N^\ell:=\phi_{\ell}(cN^{-1/d}\mathcal{C}\cap B^d(r_\ell N^{1/d}/c)),
$$
and
\begin{equation}
\omega_N^{\mathcal{C}}:=\bigcup_{\ell=1}^L\omega_N^\ell.
\end{equation}
Observing that $\rho\,  \textup{vol}(B^d(1))d\lambda_d/c^d=1$, we see from \eqref{kappa1} and  \eqref{cardCr} that  as $N\to \infty$ the cardinality of $\omega_N^{\mathcal{C}}$ satisfies:
\begin{equation}\label{cardOmegaNC}
\#\omega_N^{\mathcal{C}}=\sum_{\ell=1}^L \#(\mathcal{C}\cap B^d(r_\ell N^{1/d}/c))\ge (1+\kappa_1 \epsilon)^{-1}N(1+o(1)).
\end{equation}
Let  $\delta$ denote the smallest  distance between any pair of distinct spherical caps $C(a_\ell,r_\ell)$ and $C(a_{\ell'},r_{\ell'})$.
The {\em cross energy} for $\ell\neq \ell'$ satisfies
\begin{equation}\label{crossen}
\frac{E_{h_N}(\omega_N^\ell,\omega_N^{\ell'})}{N}:=\frac{1}{N} \sum_{\substack{x\in \omega_N^\ell\\ y\in \omega_N^{\ell'}}}h_N(\langle x,y \rangle)\le N\exp\left(-\frac{\alpha\delta^2}{c^2}N^{2/d}\right)=o(1),
\end{equation}
as $N\to \infty$.

Using \eqref{phiineq} and defining $\alpha_\epsilon=\alpha(1- \kappa_2\epsilon)^2$, we obtain
\begin{equation*}
\begin{split}
E_{h_N}(\omega_N^\ell)&=\sum_{\substack{x, y\in \mathcal{C}\cap B^d(r_\ell N^{1/d}/c)\\ x\neq y}}\exp({-\alpha\frac{|\phi_\ell(cN^{-1/d} x)-
\phi_\ell(cN^{-1/d}y)|^2}{(cN^{-1/d})^2}})\\
&\le \sum_{\substack{x, y\in \mathcal{C}\cap B^d(r_\ell N^{1/d}/c)\\ x\neq y}}\exp(-\alpha(1- \kappa_2\epsilon)^2|x-y|^2)=E_{f_{\alpha_\epsilon}}(\mathcal{C}\cap B^d(r_\ell N^{1/d}/c))\\
&\le  [\rho\,  \textup{vol}(B^d(1))]\,(r_\ell^d N/c^d) E_{f_{\alpha_\epsilon}}(\mathcal{C})(1+o(1))= \frac{Nr_\ell^d}{d\lambda_d}E_{f_{\alpha_\epsilon}}(\mathcal{C})(1+o(1)).
\end{split}
\end{equation*}
Using the above estimate for $E_{h_N}(\omega_N^\ell)$ together with \eqref{cardOmegaNC} and \eqref{crossen} we obtain as $N\to\infty$,
\begin{equation}\label{alltogether}
\begin{split}
 \frac{\mathcal{E}_{h_{N}}(\mathbb{S}^d,\#\omega_N^{\mathcal{C}})}{\#\omega_N^{\mathcal{C}}}&\leq\frac{E_{h_N}(\omega_N^{\mathcal{C}})}{\#\omega_N^{\mathcal{C}}}\le (1+\kappa_1\epsilon)\sum_{\ell=1}^L \frac{E_{h_N}(\omega_N^\ell)}{N}(1+o(1))  +o(1)\\
&\le (1+\kappa_1\epsilon) \frac{1}{d\lambda_d}\left(\sum_{\ell=1}^L r_\ell^d\right)E_{f_{\alpha_\epsilon}}(\mathcal{C})(1+o(1)) +o(1)\\&\le (1+\kappa_1 \epsilon)^2E_{f_{\alpha_\epsilon}}(\mathcal{C})(1+o(1)) +o(1),
\end{split}
\end{equation}
Taking the limit inferior as $N\to\infty$ and then $\epsilon\to 0$ in \eqref{alltogether} and using \eqref{GaussSphasymp} completes the proof. \end{proof}

\section{Numerics}\label{section.ulbnum}

Translated into packing density and using Corollary \ref{thm.jacobi zeros}, inequality \eqref{eq.spheresep} provides an alternate proof of the following best-packing bound of Levenshtein \cite{LevPacking}:

\begin{corollary}
\[\Delta_d\leq \frac{z_{1}^d}{\Gamma(d/2+1)^24^d} =:L_d\]
\label{packingbound}
\end{corollary}
As $s\to\infty$, the series in $A_{s,d}$ is dominated by the first term $z_{1}^{-s}$ and using the asymptotics of $C_{s,d}$ in (\ref{eq.Csinfty}), we see that
\begin{equation}
    \lim_{s\to\infty}\bigg[\frac{C_{s,d}}{A_{s,d}}\bigg]^{1/s} = \bigg[\frac{L_d}{\Delta_d}\bigg]^{1/d}=: B_d \geq 1
    \label{eq.Bd}
\end{equation}
The following table shows the values of $B_d$ in dimensions $d=1,2,3,8,$ and $24$ where $\Delta_d$ is known precisely. For $d=4,5,6,7$ where $\Delta_d$ is conjectured to be given by lattice packings, the table provides an upper bound for $B_d$.

\begin{table}[h!]
\begin{center}
\caption{Upper Bounds on $B_d$}

\begin{tabular}{|c|c|}
\hline
$d$ & $B_d$ \\
\hline
1 & 1\\
2 & 1.00589479\\
3 & 1.02703993\\
4 & 1.02440844\\
5 & 1.03861371\\
6 & 1.03461793\\
7 & 1.03156355\\
8 & 1.01742074\\
24 & 1.02403055\\

\hline
\end{tabular}
\label{tab.Bd}
\end{center}
\end{table}

\begin{figure}
\includegraphics[scale = 0.48]{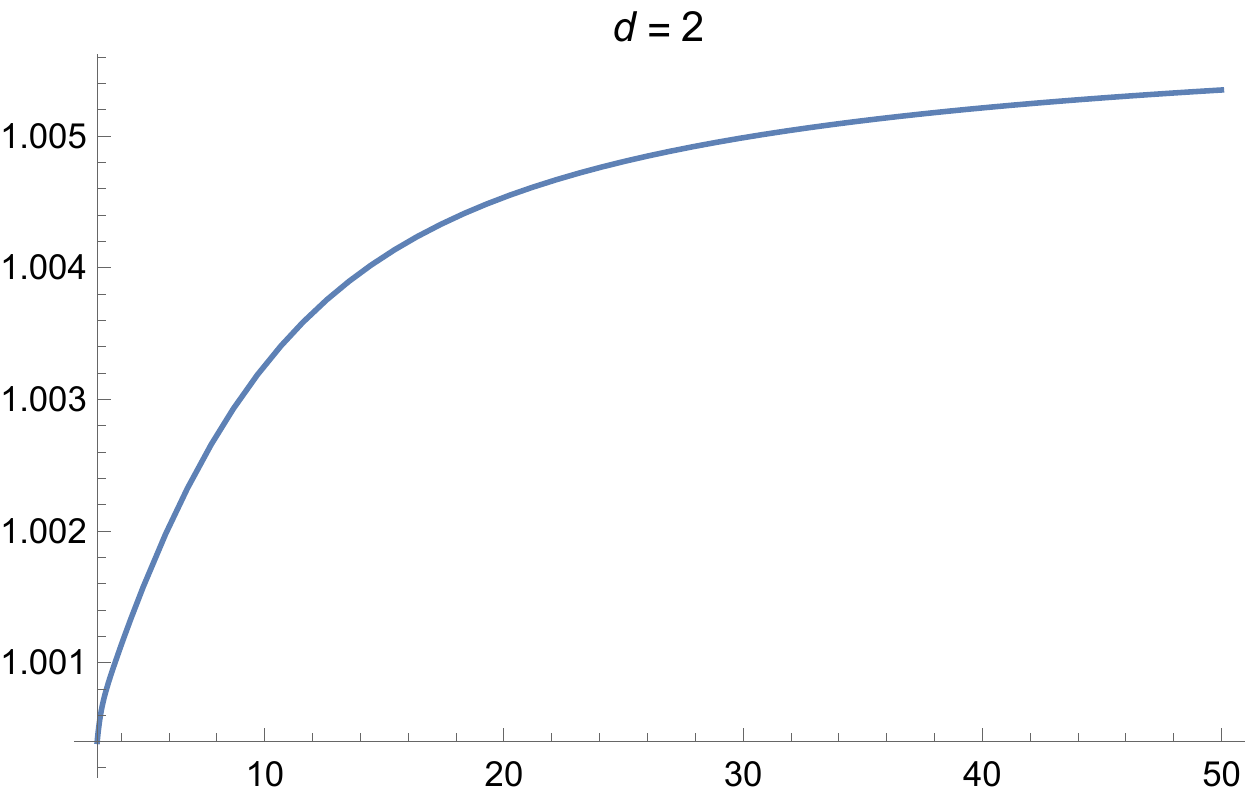} \includegraphics[scale = 0.48]{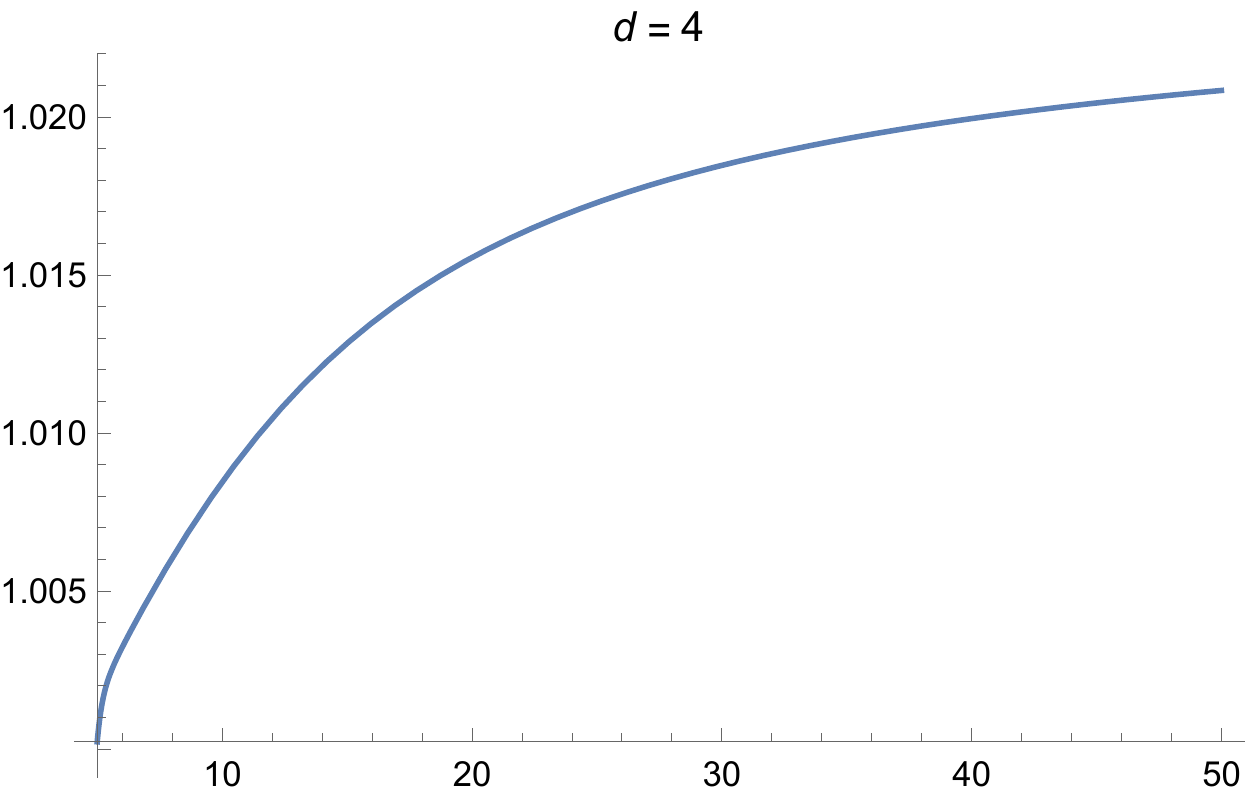}\\
\includegraphics[scale = 0.48]{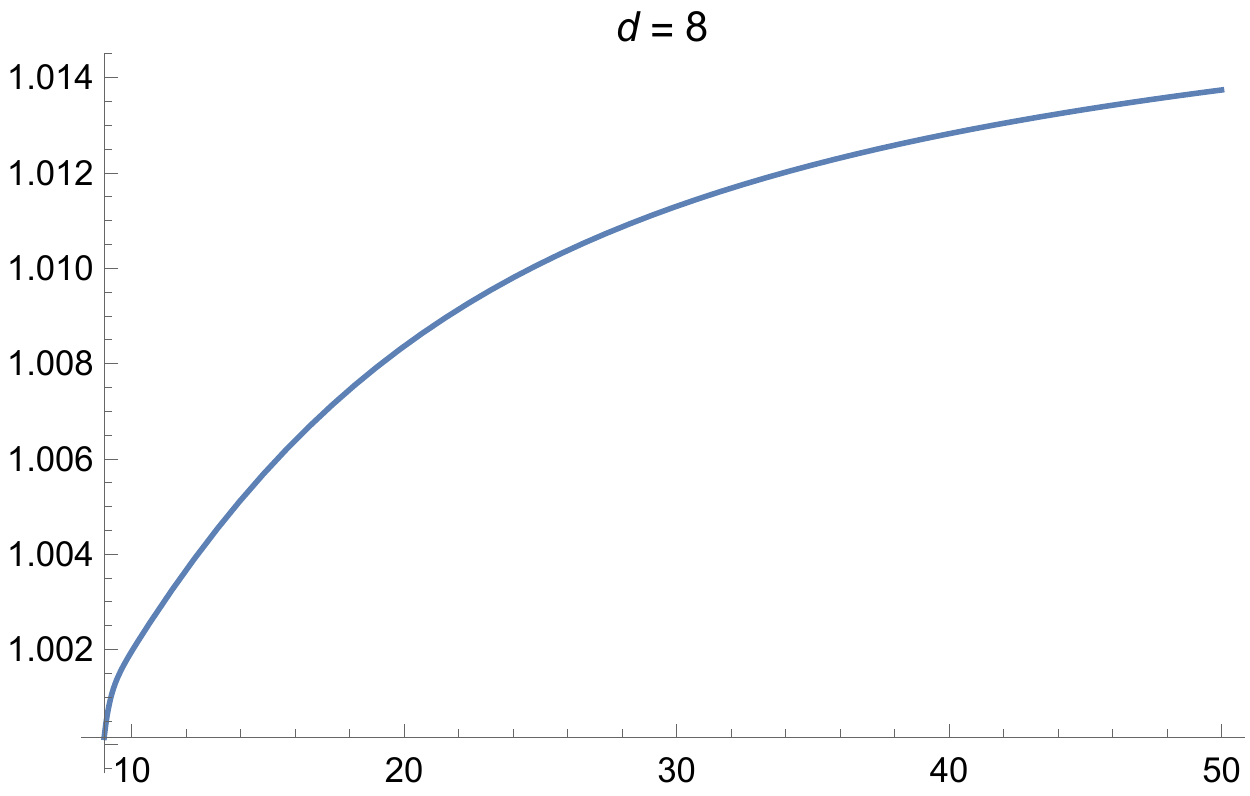} \includegraphics[scale = 0.48]{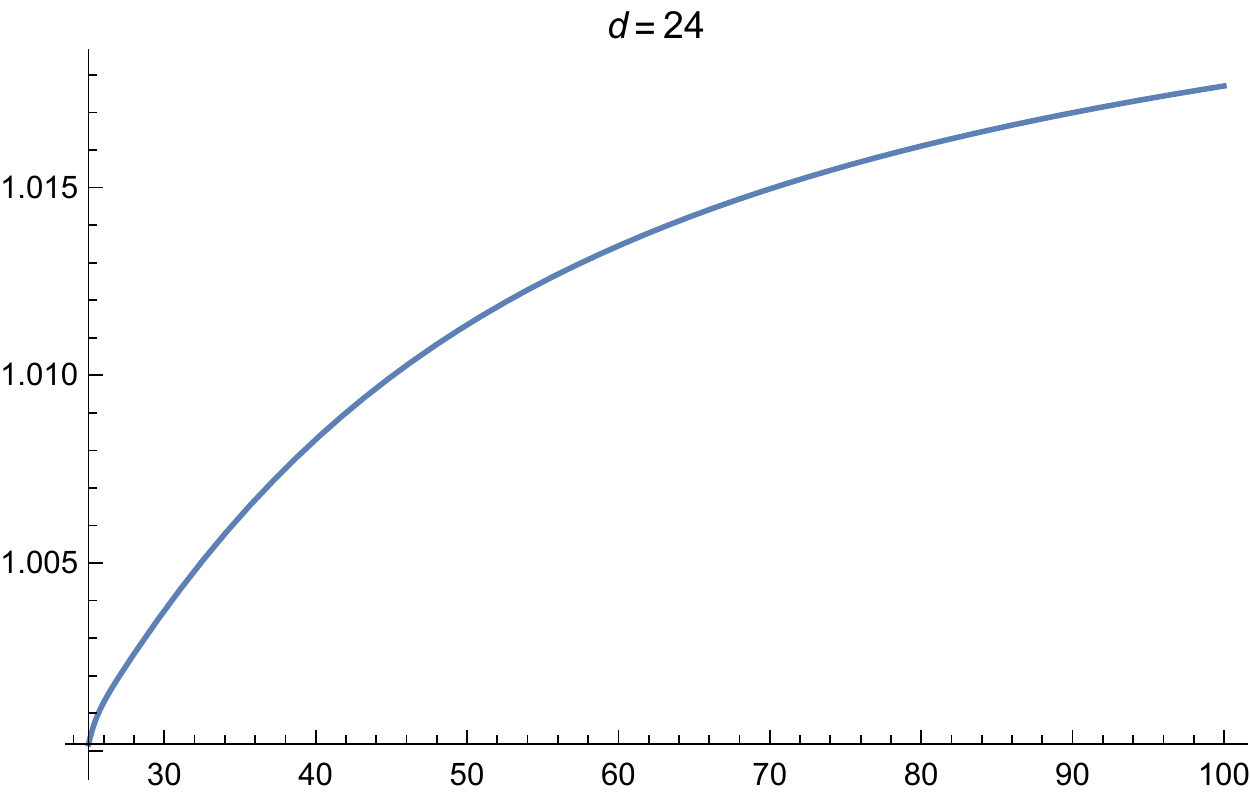}
\caption{Graphs of $f(s) = (\widetilde{C}_{s,d}/A_{s,d})^{1/s}$ for $d=2,4,8$ and $24$.}
\label{Asdfigure}
\end{figure}

%\begin{figure}
%\includegraphphics[scale = 0.48]{Hexasymptotic.pdf} \includegraphphics[scale = 0.48]{D4asymptotic}\\
%\includegraphphics[scale = 0.48]{E8asymptotic.pdf} \includegraphphics[scale = 0.48]{Leechasymptotic}
%\caption{Graphs of $f(s) = (\widetilde{C}_{s,d}/A_{s,d})^{1/s}$ for $d=2,4,8$ and $24$.}
%\label{Asdfigure}
%\end{figure}

For $d = 2,4,8,$ and $24$, where $\widetilde{C}_{s,d}$ is given in Conjecture \ref{Csdconj} we plot
\[f(s):=\bigg[\frac{\widetilde{C}_{s,d}}{A_{s,d}}\bigg]^{1/s}.\]
The Epstein zeta functions for the $D_4$, $E_8$, and Leech lattices are calculated using known formulas for the theta functions (see \cite[Ch. 4]{Conway})
\[ \Theta_\Lambda(z) = \sum_{x\in\Lambda}e^{i\pi z|x|^2},\ \ \ \ \ \ \ \ \ \ \ \mathrm{Im}\,z > 0.\]
Since these three lattices have vectors whose squared norms are even integers, we let $q=e^{i\pi z}$ and write
\[\Theta_{\Lambda_d}(z) = \sum_{m=1}^\infty N_d(m) q^{2m},\]
where $N_d(m)$ counts the number of vectors in $\Lambda_d$, $d=4,8,24$ of squared norm $2m$. Thus the Epstein zeta function
\[\zeta_{\Lambda_d}(s) = \sum_{m=1}^\infty \frac{N_d(m)}
{(2m)^{s/2}}.\]
For the $D_4$ lattice, a classical result from number theory gives
\[N_4(m) = 24\sum_{\substack{d|2m,\\ d \textup{ odd}}}d.\]
For the $E_8$ lattice, we have
\[N_8(m) =  240 \sigma_3(m),\]
where
\[\sigma_k(m) = \sum_{d|m}d^k \]
is the divisor function. Finally for the Leech lattice, it is known that
\[N_{24}(m) = \frac{65520}{691}  \left(\sigma_{11} (m) - \tau (m) \right), \]
where $\tau(m)$ is the Ramanujan tau function defined in \cite{Ramtau}.

Figure \ref{Asdfigure} plots $f(s)$ for $d=2,4,8$ and $24$. In these dimensions the graphs monotonically increase to the limit $B_d$ as $s\to \infty$ and decrease to 1 as $s\to d^+$, demonstrating Proposition \ref{thm.2tight}.

We remark that in high dimensions, it is likely that lattice packings are no longer optimal and less is known or conjectured regarding $C_{s,d}$. The Levenshtein packing bound from Corollary \ref{packingbound} yields
for large $d,$
\begin{equation}
\Delta_d\leq 2^{-0.5573d}
\label{asymppacking}
\end{equation}
and thus
\[B_d = O\bigg(\frac{2^{-0.5573}}{\Delta_d^{1/d}}\bigg).\]\

\noindent\emph{Acknowledgment.} The authors are grateful to J. S. Brauchart for his helpful suggestions.

 \bibliographystyle{plain}
 \bibliography{ULBasymptoticsbibliography}

\end{document}